\documentclass[11pt]{article}

\usepackage[margin=1in]{geometry}
\usepackage{amsmath,amssymb,amsthm,mathtools}
\usepackage{booktabs}
\usepackage{enumitem}
\usepackage{microtype}
\usepackage[numbers,sort&compress]{natbib}
\usepackage{xcolor}
\usepackage[colorlinks=true,linkcolor=blue!55!black,citecolor=blue!55!black,urlcolor=blue!55!black]{hyperref}
\usepackage[nameinlink,noabbrev]{cleveref}

\newcommand{\bits}{\{0,1\}}
\newcommand{\Unif}{\operatorname{Unif}}
\newcommand{\supp}{\operatorname{supp}}
\newcommand{\KL}{D_{\mathrm{KL}}}

\newcommand{\cF}{\mathcal{F}}

\newcommand{\E}{\mathbb{E}}
\newcommand{\Prb}{\mathbb{P}}

\newtheorem{theorem}{Theorem}[section]
\newtheorem{lemma}[theorem]{Lemma}

\theoremstyle{definition}
\newtheorem{definition}[theorem]{Definition}
\theoremstyle{remark}

\title{Beyond the Static Barrier for Ordinary Dynamic Approximate Membership}
\author{
  Qizhi Chen\\[2pt]
  \small Independent Researcher\\[-1pt]
  \small\href{mailto:15312329@qq.com}{\texttt{15312329@qq.com}}
  \and
  Zhebei Shen\\[2pt]
  \small Zhejiang University\\[-1pt]
  \small\href{mailto:shenzhebei@zju.edu.cn}{\texttt{shenzhebei@zju.edu.cn}}
  \and
  Zhehan Yu\\[2pt]
  \small Independent Researcher\\[-1pt]
  \small\href{mailto:zhehanyu0630@gmail.com}{\texttt{zhehanyu0630@gmail.com}}
}
\date{}

\begin{document}
\maketitle

\begin{abstract}
We prove a strict space separation between static and ordinary dynamic
approximate membership at every fixed error rate.  For each fixed
$\varepsilon\in(0,1)$, a capacity-$n$ ordinary dynamic filter over a universe
of size $u$, with zero false negatives, pointwise false-positive probability at
most $\varepsilon$, arbitrary history dependence, a free public random tape,
and at most $H$ bits of persistent state, satisfies
\[
  H\ge
  \bigl(\log_2(1/\varepsilon)+a_\varepsilon^{\rm c}\bigr)n-o(n),
\]
under only $u/n\to\infty$.  The constant $a_\varepsilon^{\rm c}$ is an
explicit variational threshold obtained by preserving the dependence between
the parent accepted mass and the successor reservoir.

The structural step is a common-continuation transport lemma.  A joint
posterior KL bound gives a branch-specific survivor support; the same legal
delete--insert word transports that support to one successor state, forcing an
accepted reservoir.  We then keep the parent outside mass $1-X$ in the
conditional-entropy argument instead of replacing it by $1-\varepsilon$.
This yields a two-variable analytic envelope, with no selected thresholds,
dyadic witnesses, or numerical assumptions.

\end{abstract}

\section{Introduction}

Approximate membership asks for a compact representation of a set
$S\subseteq U$ that never rejects a member and accepts a nonmember with
probability at most $\varepsilon$.  The static counting bound of Carter et al.
\cite{CarterEtAl1978} is
\[
  n\log_2(1/\varepsilon)-o(n)
\]
when $u/n\to\infty$.  Dynamic filters must preserve this guarantee under
legal insertions and deletions while retaining only their persistent state.
The difficulty is that one state may represent many current sets, whereas the
same update word must be legal for every set in the corresponding fiber.

The small-error regime is now understood sharply: for $\varepsilon=o(1)$,
Kuszmaul, Liang, and Zhou \cite{KuszmaulLiangZhou2025} prove the first-order
bound $n\log_2(1/\varepsilon)+n\log_2 e-o(n)$, building on earlier work of
Lovett and Porat \cite{LovettPorat2013} and Kuszmaul and Walzer
\cite{KuszmaulWalzer2024}.  The constant-error regime is different.  Their
model allows arbitrary history dependence, a free public random tape, and
pointwise error after every fixed legal history; we work in precisely this
model and retain the natural-universe assumption $u/n\to\infty$.
For the ordinary legal-update semantics at fixed error, the Carter bound had
remained the general baseline; previous strict separations either used related
incremental semantics or imposed additional structure.

\subsection{Our result}

For $\varepsilon\in(0,1)$, let $H^*_\varepsilon(n,u)$ denote the infimum
worst-case number of persistent bits used by such a filter of capacity $n$.
Write $\ell=\log_2(1/\varepsilon)$ for the Carter baseline.  We measure space
above this baseline by writing a candidate bound as
$H\le(\ell+a)n+o(n)$; thus $a\ge0$ is an excess-space budget.  The coupled
argument rules out every budget below a threshold $a_\varepsilon^{\rm c}$.
Formally, this threshold is the largest $a$ for which the successor-reservoir
cost still exceeds the available budget; the precise envelope is defined in
\Cref{sec:variational}.  Intuitively, it balances the information needed to
preserve possible parent keys through a replacement against the entropy needed
to represent fresh keys accepted by the successor.

\begin{theorem}[Main theorem: coupled fixed-error separation]\label{thm:general-epsilon}
For every fixed $\varepsilon\in(0,1)$,
\[
 H^*_\varepsilon(n,u)
 \ge
 \bigl(\log_2(1/\varepsilon)+a_\varepsilon^{\rm c}\bigr)n-o(n)
\]
along every sequence with $n\to\infty$ and $u/n\to\infty$.
\end{theorem}

In other words, the dynamic requirement costs a positive linear excess over
the static Carter term, and $a_\varepsilon^{\rm c}$ is the amount of excess
that this one-step argument certifies.

The theorem allows arbitrary history dependence, a free public random tape,
and pointwise error after every legal history.  It assumes neither stationarity
nor canonical states, monotonicity, or a fingerprinting implementation.  The
threshold is defined by a strict variational inequality, so no numerical
witness or attained optimizer is part of the statement.

\subsection{Position relative to prior work}

Lovett and Porat \cite{LovettPorat2013} proved a strict separation in a related
incremental model.  Their hard source allows insertion words with repeated
labels, whereas the ordinary API here requires every insertion to be fresh;
continuation legality is therefore a genuine issue.  The present theorem keeps
the ordinary legal-history semantics and loses only $o(n)$ bits from the
natural-universe comparison.  Other dynamic-filter constructions and entropy
rate results use different resource conventions, such as eventual-size
resizability or high-probability current-state space
\cite{PaghSegevWieder2013,BenderEtAl2018,BerceaEven2020,BlellochEtAl2026}.

More broadly, resource-constrained dynamic systems also include compact
entropy-encoded representations and exact optimization of capacity-constrained
network paths \cite{BlellochEtAl2026,LuoShen2026}.  These works provide
constructions or exact algorithms in their respective settings; our focus is
the complementary information-theoretic cost of maintaining compact state
under updates.

\subsection{Technical overview}

The proof has one structural idea: all losses are charged on average in their
natural information units.  In Carter's argument, an accepted-set density
$X$ costs the exact tangent energy
\[
 \mathcal E_\varepsilon(X)
 =\log_2\frac\varepsilon X
  +\frac{X-\varepsilon}{\varepsilon\ln2}.
\]
A joint relative-entropy chain rule charges the conditional posterior cost
$\kappa$ of a simultaneous replacement of $\lambda n$ keys to the same
budget.  The common update word forces a successor state to retain at least
$X2^{-\kappa/(1-\lambda)}u-o(u)$ accepted parent keys.  Minimizing this
quantity at fixed total energy produces the survivor profile
$\phi_{\varepsilon,\lambda}$.  The coupled argument also keeps
$m=\mathbb E[X]/\varepsilon$: an exact logarithmic identity and Jensen give the
additional lower bound $\psi_\lambda(m,a)$ in \eqref{eq:coupled-survivor}.

In the opposite direction, let $C$ be the successor's accepted reservoir
outside the parent accepted set.  Conditional entropy of the fresh keys gives
directly
\[
 \frac Hn\ge
 \lambda(1-\varepsilon m)
 \log_2\frac{1-\varepsilon m}{\E|C|/u}-o(1).
\]
The parent outside mass is therefore coupled to the survivor lower bound,
rather than replaced by $\beta$.  Combining the coupled inequalities and
optimizing over $\lambda$ and the feasible mean $m$ yields
\eqref{eq:coupled-envelope}.  The only universe-dependent errors are
$O(n/u)$ and $O(n^{-1/2})$, so $u/n\to\infty$ suffices.

\subsection{Organization}

\Cref{sec:model} defines the model and the replacement experiments.
\Cref{sec:carter} extracts the exact Carter energy.
\Cref{sec:replacement} proves the joint posterior replacement lemma, and
\Cref{sec:forcing} turns posterior thickness into a bound on the average
successor reservoir.  \Cref{sec:integrated-cover} proves the reverse
conditional-entropy bound.  \Cref{sec:variational,sec:general-epsilon}
derive the survivor profiles, establish the coupled variational threshold, and
prove the main theorem.  The appendix records the public-tape and rounding
details needed for the asymptotic quantifiers.

\section{Model and experiments}\label{sec:model}

All logarithms and information quantities are in bits.  Let $U$ be a finite
universe of size $u$, and let $n$ be the capacity.  We fix once and for all a
total order on $U$.

\begin{definition}[Ordinary dynamic filter]\label{def:filter}
An ordinary dynamic approximate-membership filter has a persistent state in
$\bits^H$ and supports the operations
\(\mathsf{Insert}(x)\), \(\mathsf{Delete}(x)\), and
\(\mathsf{Query}(x)\).  A history is legal if an insertion is applied only to
a key outside the current set, a deletion only to a key inside the current
set, and the set size always lies in $[0,n]$.

All random bits used by initialization, updates, and queries are sampled in
advance as a read-only public tape $R$, independent of the history.  Once
$R=r$ is fixed, every operation is deterministic.  Formally, an update at
time $t$ is a transition of the form
\[
  \Delta^{\mathsf{op}}_{r,t}:\bits^H\times U\longrightarrow\bits^H,
  \qquad \mathsf{op}\in\{\mathsf{Insert},\mathsf{Delete}\},
\]
and a query is a function of $(r,t,m,x)$.  Thus history dependence is
unrestricted but can act only through the $H$-bit persistent state (together
with $r$, $t$, and the current operation key); the algorithm has no uncharged
access to the preceding transcript.  The tape is free; only the $H$
persistent bits are charged.  At operation time $t$, write
$A_t(r,m)\subseteq U$ for the keys accepted by the query algorithm on tape
$r$ and state $m$.  Allowing the query map to depend on $t$ only strengthens
the model.  Fresh query coins, if desired, are addressed on $R$ by the query
invocation and key.  Fixing $r$ therefore couples all counterfactual queries
into the single accepted set $A_t(r,m)$; no independence across keys is used.

For every tape, every legal history $h$ with current set $S(h)$, and every
$x\in S(h)$, the answer is \textsc{yes}.  For every fixed legal history $h$
and every fixed $x\notin S(h)$,
\[
  \Prb_R[x\in A_{|h|}(R,M_R(h))]\le\varepsilon.
  \tag{2.1}\label{eq:fpr}
\]
No restriction is placed on running time, relocation, or auxiliary
representations within the charged state.  All algorithm maps are assumed
measurable; equivalently, one may take $R$ to be the standard product space
of independent public coin flips.
\end{definition}

The explicit time subscript in $A_t$ prevents any implicit use of a free
clock.  It may be dropped when the relevant time is clear.

\subsection{The parent experiment}

Choose
\[
  S\sim\Unif\binom Un
\]
independently of $R$, start from the empty set, and insert the elements of $S$
in the order induced by the fixed order on $U$.  Let $M$ be the state at time
$t_0=n$, set $Z=(R,M)$, and abbreviate
\[
  A_Z=A_{t_0}(R,M),\qquad a_Z=|A_Z|.
\]
Zero false negatives imply $S\subseteq A_Z$ on every tape.

For each fixed $r$, the canonical build defines a map
$S\mapsto M_r(S)$.  Hence for every nonempty fiber
\[
  \cF_{r,m}=\left\{s\in\binom Un:M_r(s)=m\right\},
\]
the conditional source law is the explicit finite distribution
\[
  \mu_{r,m}=\Unif(\cF_{r,m}).
\]
This definition does not require the singleton event $R=r$ to have positive
probability.  Statements involving $Z=z=(r,m)$ below mean these fiberwise
statements, integrated over $r$.

\subsection{The replacement experiment}

Fix a constant replacement fraction $0<\lambda<1$ and put
\[
 q=\lfloor\lambda n\rfloor,\qquad k=n-q.
\]
After the parent build, independently choose
\[
  D\mid S\sim\Unif\binom Sq,
  \qquad
  I\mid S\sim\Unif\binom{U\setminus S}q.
  \tag{2.2}\label{eq:replacement-law}
\]
Delete the elements of $D$ and then insert the elements of $I$, using the
fixed universe order within each batch.  The word is legal and the successor
set is
\[
  S'=(S\setminus D)\cup I.
\]
Let $M'$ be the successor state at time $t_1=n+2q$, and write
\[
  Z'=(R,M'),\qquad A'_{Z'}=A_{t_1}(R,M'),\qquad a'_{Z'}=|A'_{Z'}|.
\]

\begin{lemma}[Tape-independent uniform endpoints]\label{lem:experiments}
The complete random history in both experiments is independent of $R$.
Moreover, both $S$ and $S'$ are uniform in $\binom Un$.
\end{lemma}

\begin{proof}
The sampling rules and the canonical order never inspect $R$.  Uniformity of
$S$ is by definition.  The law of $S'$ is invariant under every permutation
of $U$ and is supported on the $n$-subsets of $U$; therefore it is uniform.
Equivalently, the number of pairs $(S,D,I)$ producing a fixed $S'$ is the same
for every $S'$.
\end{proof}

The successor state $M'$ may depend on the full history $(S,D,I)$ and need not
be determined by $S'$ alone.  Nevertheless, $S'$ is independent of $R$ and
$M'$ has at most $2^H$ values, so
$I(S';R,M')=I(S';M'\mid R)\le H$.  This is all that is needed when the Carter
argument is applied to the successor experiment; zero false negatives again
support the posterior of $S'$ on $\binom{A'_{Z'}}n$.

The following averaging consequence will be used twice.

\begin{lemma}[Accepted-size average]\label{lem:fpr-average}
For the parent and successor experiments, respectively,
\[
  \E|A_Z|\le n+\varepsilon(u-n),
  \qquad
  \E|A'_{Z'}|\le n+\varepsilon(u-n).
  \tag{2.3}\label{eq:accepted-average}
\]
\end{lemma}

\begin{proof}
Fix any realized history before exposing $R$.  Sum \eqref{eq:fpr} over its
$u-n$ nonmembers and use zero false negatives for its $n$ members.  Averaging
over the tape and then over the independently sampled history proves the
claim.
\end{proof}

\section{The exact Carter energy}\label{sec:carter}

Write
\[
  \ell=\log_2(1/\varepsilon),\qquad
  \beta=1-\varepsilon,\qquad
  \delta=n/u,
\]
and define the tangent gap
\[
 \mathcal E_\varepsilon(x)
 =\log_2\frac\varepsilon x
  +\frac{x-\varepsilon}{\varepsilon\ln2},
 \qquad 0<x\le1.
 \tag{3.1}\label{eq:carter-energy}
\]
This function is nonnegative and vanishes only at $x=\varepsilon$.

For real $x>n-1$, let
\[
 f(x)=\log_2\binom xn
 :=\sum_{j=0}^{n-1}\log_2(x-j)-\log_2(n!).
\]
Conditioned fiberwise on $Z=z$, define the posterior deficit
\[
 d_z
 =\KL\left(\mu_z\middle\|\Unif\binom{A_z}{n}\right)
 =f(a_z)-H(S\mid Z=z).
 \tag{3.2}\label{eq:deficit}
\]
The first distribution is supported on the second because every represented
source set is accepted.

\begin{lemma}[Integrated Carter bound]\label{lem:carter}
For $X=a_Z/u$,
\[
 \E\mathcal E_\varepsilon(X)+\frac{\E d_Z}{n}
 \le \frac Hn-\ell+O_\varepsilon(\delta).
 \tag{3.3}\label{eq:integrated-carter}
\]
In particular, the accepted-set deviation and the posterior deficit spend the
same additive slack above the Carter bound.
\end{lemma}

\begin{proof}
Put $\Lambda_n=\log_2\binom un$.  Since $S$ is independent of $R$,
\[
 H\ge I(S;R,M)=\Lambda_n-\E f(a_Z)+\E d_Z.
 \tag{3.4}\label{eq:carter-decomp}
\]
For every $a\in[n,u]$, the product formula gives
\[
 \Lambda_n-f(a)
 =\sum_{j=0}^{n-1}\log_2\frac{u-j}{a-j}
 \ge n\log_2\frac ua.
\]
Consequently,
\[
 \frac Hn\ge\E[-\log_2X]+\frac{\E d_Z}{n}.
 \tag{3.5}\label{eq:carter-pointwise}
\]
By \Cref{lem:fpr-average},
$\E X\le\varepsilon+\beta\delta$.  Substituting
\eqref{eq:carter-energy} into \eqref{eq:carter-pointwise} therefore yields
\[
 \E\mathcal E_\varepsilon(X)+\frac{\E d_Z}{n}
 \le\frac Hn-\ell+\frac{\beta\delta}{\varepsilon\ln2},
\]
which is \eqref{eq:integrated-carter}.
\end{proof}

\section{The joint posterior replacement lemma}\label{sec:replacement}

Fix a parent fiber value $z=(r,m)$ and abbreviate $A=A_z$, $a=|A|$.
Let
\[
  \mu_z=\Unif(\cF_{r,m}),
  \qquad
  \nu_z=\Unif\binom An.
\]
Starting from either source law, apply the same kernel: sample
$D\in\binom Sq$ uniformly, then $I\in\binom{U\setminus S}q$ uniformly, and
output $K=S\setminus D$.  Denote the resulting joint laws of $(D,I,K)$ by
$P_z$ and $Q_z$.

For an actual branch $(z,d,i)$ of positive probability, define
\[
  V_{z,d,i}
  =\bigcup\supp P_z(K\mid D=d,I=i),
  \qquad v_{z,d,i}=|V_{z,d,i}|.
  \tag{4.1}\label{eq:V-def}
\]

\begin{lemma}[Joint replacement chain rule]\label{lem:joint-kl}
For every parent fiber,
\[
 \E_{(D,I)\sim P_z}
 \KL\left(P_z(K\mid D,I)\middle\|Q_z(K\mid D,I)\right)
 \le d_z.
 \tag{4.2}\label{eq:conditional-kl}
\]
For every actual branch $(d,i)$,
\[
  Q_z(K\mid D=d,I=i)
  =\Unif\binom{A\setminus(d\cup i)}k.
  \tag{4.3}\label{eq:reference-conditional}
\]
In particular, if
\[
  b_{z,d,i}=a-q-|i\cap A|,
  \tag{4.4}\label{eq:b-def}
\]
then
\[
 \KL\left(P_z(K\mid d,i)\middle\|Q_z(K\mid d,i)\right)
 \ge\log\frac{\binom{b_{z,d,i}}k}{\binom{v_{z,d,i}}k}
 \ge k\log\frac{b_{z,d,i}}{v_{z,d,i}}.
 \tag{4.5}\label{eq:support-kl}
\]
\end{lemma}

\begin{proof}
If $P_z(d,i)>0$, an actual compatible source in $\cF_{r,m}$ also has positive
probability under $\nu_z$; hence $P_z(D,I)\ll Q_z(D,I)$ and the reference
conditional law below is defined on every actual branch.
Applying a common stochastic kernel cannot increase relative entropy, so
\[
  \KL(P_z(D,I,K)\|Q_z(D,I,K))
  \le\KL(\mu_z\|\nu_z)=d_z.
\]
The chain rule expands the left-hand side as
\[
 \KL(P_z(D,I)\|Q_z(D,I))
 +\E_{(D,I)\sim P_z}
   \KL(P_z(K\mid D,I)\|Q_z(K\mid D,I)).
\]
Dropping the first nonnegative term gives \eqref{eq:conditional-kl}.  Notice
that this expectation uses the \emph{actual} branch marginal $P_z(D,I)$.

Under the reference law, every feasible triple $(d,i,k_0)$ has probability
\[
  \frac1{\binom an\binom nq\binom{u-n}q},
\]
because it uniquely specifies the source $d\cup k_0$.  Thus, after fixing
$(d,i)$, all $k$-subsets of $A\setminus(d\cup i)$ are equally likely,
proving \eqref{eq:reference-conditional}.  Since $d\subseteq A$,
$|d|=q$, and $d\cap i=\varnothing$, this ground set has size
$a-q-|i\cap A|$.

The actual conditional law is supported on $k$-subsets of $V=V_{z,d,i}$.
Therefore its entropy is at most $\log\binom vk$, whereas the reference law is
uniform on $\binom bk$.  This proves the first inequality in
\eqref{eq:support-kl}.  The second follows from
\[
  \frac{\binom bk}{\binom vk}
  =\prod_{j=0}^{k-1}\frac{b-j}{v-j}
  \ge\left(\frac bv\right)^k.
\]
\end{proof}

The relative entropy in \eqref{eq:conditional-kl} is deliberately taken only
after jointly revealing both batches.  We neither assert nor need that the
reference marginal $Q_z(I\mid D)$ is uniform.

For an actual branch, normalize its conditional posterior cost as
\[
 \kappa
 =\frac1n\KL\left(P_Z(K\mid D,I)\middle\|Q_Z(K\mid D,I)\right).
 \tag{4.6}\label{eq:kappa}
\]
The chain rule and \eqref{eq:support-kl} give
\[
 \E\kappa\le\frac{\E d_Z}{n},
 \qquad
 \frac{v_{Z,D,I}}u
 \ge(X-2\lambda\delta)2^{-\kappa/(1-\lambda)}.
 \tag{4.7}\label{eq:average-kappa-thickness}
\]
Indeed, $k\ge(1-\lambda)n$ and
$b=a-q-|I\cap A|\ge a-2q\ge a-2\lambda n$.  Thus
$\log_2(b/v)\le n\kappa/k\le\kappa/(1-\lambda)$.

\section{A common successor and its outside reservoir}\label{sec:forcing}

We now use the dynamic semantics.  Fix an actual branch $(z,d,i)$ of positive
probability.  Let $m_i$ denote the state obtained, on tape $r$, by applying
the canonical word
\[
  \mathsf{Delete}(d)\ ;\ \mathsf{Insert}(i)
\]
to the parent state $m$.  The notation suppresses the dependence on
$(z,d)$.

\begin{lemma}[Common-word forcing]\label{lem:common-word}
Let $A'=A_{t_1}(r,m_i)$.  Then
\[
  V_{z,d,i}\cup i\subseteq A'.
  \tag{5.1}\label{eq:common-forcing}
\]
Moreover, $V_{z,d,i}\subseteq A$, and hence
\[
  |A'\setminus A|\le |A'|-v_{z,d,i}.
  \tag{5.2}\label{eq:outside-capacity}
\]
\end{lemma}

\begin{proof}
For every $K$ in the conditional support in \eqref{eq:V-def}, the set
$d\cup K$ is represented by the same parent state $m$ on tape $r$.  The
branch is compatible, so $d\subseteq d\cup K$ and
$i\cap(d\cup K)=\varnothing$.  Thus the same delete--insert word is legal for
every such source.  Fixed-tape determinism sends all of them from $m$ to the
same state $m_i$.  Their successor sets are $K\cup i$, which must be accepted
by zero false negatives.  Taking the union over $K$ proves
\eqref{eq:common-forcing}.  Every compatible $K$ is a subset of a parent
source represented at $m$, so zero false negatives also give $V\subseteq A$.
It follows that $V\subseteq A\cap A'$, proving
\eqref{eq:outside-capacity}.
\end{proof}

\begin{lemma}[Average successor reservoir]\label{lem:average-reservoir}
Let
\[
 Y=|A'_{Z'}|/u,\qquad C=A'_{Z'}\setminus A_Z,\qquad c=|C|/u.
\]
Then, branchwise,
\[
 c\le Y-(X-2\lambda\delta)2^{-\kappa/(1-\lambda)},
 \tag{5.3}\label{eq:branch-reservoir}
\]
and consequently
\[
 \E c\le
 \varepsilon-\E[X2^{-\kappa/(1-\lambda)}]+O_\lambda(\delta).
 \tag{5.4}\label{eq:average-reservoir}
\]
\end{lemma}

\begin{proof}
By \Cref{lem:common-word},
$|C|\le |A'_{Z'}|-v_{Z,D,I}$.  The second inequality in
\eqref{eq:average-kappa-thickness} proves \eqref{eq:branch-reservoir}.
Averaging and using $\E Y\le\varepsilon+\beta\delta$ from
\Cref{lem:fpr-average} proves \eqref{eq:average-reservoir}.
\end{proof}

\section{Integrated replacement covering}\label{sec:integrated-cover}

The next lemma counts no selected family of branches.  It reads the required
state information directly from the random insertion batch.

\begin{lemma}[Reservoir entropy]\label{lem:reservoir-entropy}
Let
\[
 p=1-X,\qquad T=|I\setminus A_Z|,\qquad s=T/q,
 \qquad c=|A'_{Z'}\setminus A_Z|/u.
\]
Then
\[
 H\ge q\,\E\left[s\log_2\frac pc\right].
 \tag{6.1}\label{eq:raw-reservoir-entropy}
\]
Moreover, if $\E c\le c_0+o(1)$ for a fixed $c_0<\beta$, then
\[
 \frac Hn\ge
 \lambda\beta\log_2\frac\beta{c_0}-o(1).
 \tag{6.2}\label{eq:integrated-cover}
\]
\end{lemma}

\begin{proof}
Condition on $W=(R,M,S,D)$, which fixes the parent accepted set $A=A_Z$.
Put $L=I\setminus A$.  Since $S\subseteq A$ and $I$ is uniform in
$\binom{U\setminus S}q$, conditional on $(W,T)$ the set $L$ is uniform in
$\binom{U\setminus A}T$.  On the other hand, $M'$ determines
$C=A'_{Z'}\setminus A$, and zero false negatives imply $L\subseteq C$.
Therefore
\begin{align*}
 H
 &\ge H(M'\mid W)
 \ge H(M'\mid W,T)
 \ge I(L;M'\mid W,T)\\
 &\ge\E\log_2\frac{\binom{pu}T}{\binom{cu}T}
 \ge q\,\E\left[s\log_2\frac pc\right].
\end{align*}
The last step uses
$\binom NT/\binom KT\ge(N/K)^T$.  We use the usual continuous conventions
when $T=0$; branches with $p=0$ have $T=0$ and contribute zero.

To average the logarithm without assuming independence, set
$d=sc/p$ when $p>0$, and set $d=0$ when $p=0$.  Since $C\subseteq U\setminus
A$, we have $c\le p$, and hence
\[
 d\le c+|s-p|.
 \tag{6.3}\label{eq:d-control}
\]
The perspective $h(a,b)=a\log_2(a/b)$ is jointly convex on the nonnegative
orthant, with its lower-semicontinuous boundary values.  Thus
\[
 \E\left[s\log_2\frac pc\right]
 =\E h(s,d)\ge h(\E s,\E d).
 \tag{6.4}\label{eq:perspective-jensen}
\]
Conditioned on $W$, the variable $T$ is hypergeometric and
\[
 \E[s\mid W]=\frac p{1-\delta},
 \qquad
 \E\left[\left|s-\frac p{1-\delta}\right|\middle|W\right]
 \le\frac1{2\sqrt q}.
\]
The accepted-size average gives $\E s\ge\beta$, while
\[
 \E|s-p|=O(q^{-1/2}+\delta)=o(1).
\]
It follows from \eqref{eq:d-control} that $\E d\le c_0+o(1)$.  Since
$c_0<\beta$, the function $h(a,b)$ is increasing in $a$ and decreasing in
$b$ in a neighborhood of $(\beta,c_0)$.  Combining
\eqref{eq:raw-reservoir-entropy}--\eqref{eq:perspective-jensen} with
$q/n=\lambda+o(1)$ proves \eqref{eq:integrated-cover}.
\end{proof}

\begin{lemma}[Coupled reservoir entropy]\label{lem:coupled-reservoir-entropy}
With the notation of \Cref{lem:reservoir-entropy}, let
\[
 m=\frac{\E X}{\varepsilon},\qquad
 \bar p=\E(1-X)=1-\varepsilon m.
\]
If $\E c\le c_0+o(1)$ and $c_0<\bar p$, then
\[
 \frac Hn\ge
 \lambda\bar p\log_2\frac{\bar p}{c_0}-o(1).
 \tag{6.5}\label{eq:coupled-reservoir}
\]
\end{lemma}

\begin{proof}
In the proof of \Cref{lem:reservoir-entropy}, the hypergeometric estimate gives
\[
 \E|s-(1-X)|=o(1).
\]
Thus $\E s=\bar p+o(1)$ and the auxiliary variable
$d=sc/(1-X)$ satisfies $\E d\le c_0+o(1)$.  Applying the perspective Jensen
inequality \eqref{eq:perspective-jensen}, and using that
$h(a,b)=a\log_2(a/b)$ is increasing in $a$ and decreasing in $b$ whenever
$0<b<a$, gives \eqref{eq:coupled-reservoir}.
\end{proof}

\section{The coupled variational theorem}\label{sec:variational}

We now define the single variational object used in the main theorem.  The
definitions are collected here so that the introduction can explain the proof
without interrupting its narrative.

Write
\[
 e(y)=\log_2\frac1y+\frac{y-1}{\ln2},
\]
and let $y(a)\in(0,1]$ be the lower solution of $e(y)=a$.  For
$0<\lambda<1$, define the scalar survivor profile
\[
 \phi_{\varepsilon,\lambda}(a)=
 \begin{cases}
  \varepsilon y(a),&a\le e(\lambda),\\[1mm]
  \varepsilon\lambda\,2^{-(a-e(\lambda))/(1-\lambda)},&a\ge e(\lambda).
 \end{cases}
 \label{eq:phi-profile}
\]
For $m\in[y(a),1]$, define the coupled survivor profile
\[
 \psi_\lambda(m,a)=
 \varepsilon\exp\left(
 \frac{m-1-\lambda\ln m-a\ln2}{1-\lambda}\right),
 \qquad
 B_{\varepsilon,\lambda}(a,m)=
 \max\{\phi_{\varepsilon,\lambda}(a),\psi_\lambda(m,a)\}.
 \label{eq:coupled-survivor}
\]
The coupled envelope is
\[
 \widetilde{\mathcal F}_\varepsilon(a)=
 \sup_{0<\lambda<1}\inf_{y(a)\le m\le1}
 \lambda(1-\varepsilon m)
 \left[\log_2\frac{1-\varepsilon m}
 {\varepsilon-B_{\varepsilon,\lambda}(a,m)}\right]_+,
 \label{eq:coupled-envelope}
\]
where a term with $B_{\varepsilon,\lambda}(a,m)\ge\varepsilon$ is interpreted
as $+\infty$.  Finally, set
\[
 a_\varepsilon^{\rm c}=
 \sup\{a\ge0:\widetilde{\mathcal F}_\varepsilon(a)>\ell+a\}.
 \label{eq:coupled-crossing}
\]

The profile $\phi$ is the consequence of the energy budget alone.  The profile
$\psi$ retains the mean accepted density $m$.  The maximum $B$ combines the two
survivor constraints, and the outer expression charges the entropy of the
fresh-key reservoir.

\begin{lemma}[Energy-loss envelope]\label{lem:energy-loss}
For each $0<\lambda<1$, the function
$\phi_{\varepsilon,\lambda}$ in \eqref{eq:phi-profile} is decreasing,
convex, and continuous on $[0,\infty)$, and it is continuously differentiable
on $(0,\infty)$.  For every $0<x\le1$ and $\kappa\ge0$,
\[
 x2^{-\kappa/(1-\lambda)}
 \ge\phi_{\varepsilon,\lambda}
 \bigl(\mathcal E_\varepsilon(x)+\kappa\bigr).
 \tag{7.1}\label{eq:energy-envelope}
\]
Consequently, for the random branch variables above,
\[
 \E[X2^{-\kappa/(1-\lambda)}]
 \ge
 \phi_{\varepsilon,\lambda}
 \bigl(\E[\mathcal E_\varepsilon(X)+\kappa]\bigr).
 \tag{7.2}\label{eq:energy-jensen}
\]
\end{lemma}

\begin{proof}
Write $x=\varepsilon y$ and fix
$a=\mathcal E_\varepsilon(x)+\kappa=e(y)+\kappa$.  Among feasible $y$,
minimizing $x2^{-\kappa/(1-\lambda)}$ is equivalent to minimizing
\[
 \ln y+\frac{(\ln2)e(y)}{1-\lambda}.
 \tag{7.3}\label{eq:energy-minimization}
\]
The derivative of \eqref{eq:energy-minimization} has the sign of
$y-\lambda$.  Although feasibility may also include an upper-branch root
$y_+(a)>1$, it cannot minimize the objective.  Since $\lambda<1$, the
objective is strictly increasing throughout the upper feasible component, so
its minimum there occurs at $y_+(a)$.  The roots $y_+(a)$ and $y(a)$ have the
same value of $e$, while $\ln y_+(a)>\ln y(a)$, so the lower root gives the
smaller objective.  The minimizer is therefore $y(a)$ with
$\kappa=0$ when
$y(a)\ge\lambda$, and $y=\lambda$ with $\kappa=a-e(\lambda)$ otherwise.
This proves \eqref{eq:energy-envelope} and the formula in
\eqref{eq:phi-profile}.

On the first branch, implicit differentiation gives
\[
 \phi_{\varepsilon,\lambda}'(a)
 =-\frac{\varepsilon y(a)\ln2}{1-y(a)},
\]
whose value increases with $a$.  At $a=e(\lambda)$ it equals
$-\varepsilon\lambda\ln2/(1-\lambda)$, which is also the derivative of the
second branch there.  The second branch is convex and decreasing.  Hence the
profile is decreasing and convex on $[0,\infty)$ and $C^1$ on $(0,\infty)$.
Applying Jensen's inequality to \eqref{eq:energy-envelope} proves
\eqref{eq:energy-jensen}.
\end{proof}

\begin{lemma}[Coupled survivor envelope]\label{lem:coupled-survivor}
Let
\[
 m=\frac{\E X}{\varepsilon},\qquad
 A=a\ln2,\qquad
 \E\bigl[\mathcal E_\varepsilon(X)+\kappa\bigr]\le a+o(1).
\]
Then $m\in[y(a),1]+o(1)$ and
\[
 \E\left[X2^{-\kappa/(1-\lambda)}\right]
 \ge \psi_\lambda(m,a)-o(1),
 \tag{7.4}\label{eq:psi-bound}
\]
where $\psi_\lambda$ is defined in \eqref{eq:coupled-survivor}.
\end{lemma}

\begin{proof}
Put $Y=X/\varepsilon$ and
\[
 T=(\ln2)\bigl(\mathcal E_\varepsilon(X)+\kappa\bigr)
 =-\ln Y+Y-1+(\ln2)\kappa.
\]
The convexity of $-\ln y+y-1$ gives
$-\ln m+m-1\le A+o(1)$, while the accepted-size estimate gives
$m\le1+o(1)$.  Thus $m\in[y(a),1]+o(1)$.
Moreover, branchwise,
\[
 \ln\frac{X2^{-\kappa/(1-\lambda)}}{\varepsilon}
 =\frac{Y-1-\lambda\ln Y-T}{1-\lambda}.
\]
Taking expectations and using $\E\ln Y\le\ln m$ and $\E T\le A+o(1)$ gives
\[
 \E\ln\frac{X2^{-\kappa/(1-\lambda)}}{\varepsilon}
 \ge
 \frac{m-1-\lambda\ln m-A}{1-\lambda}-o(1).
\]
Since $\ln$ is concave, $\ln \E Z\ge \E\ln Z$ for the positive random
variable $Z=X2^{-\kappa/(1-\lambda)}$.  Exponentiating proves
\eqref{eq:psi-bound}.
\end{proof}

\section{Proof of the fixed-error separation theorem}
\label{sec:general-epsilon}

We first verify that the coupled threshold is well defined.

\begin{lemma}[Coupled balance]\label{lem:coupled-balance}
The function $\widetilde{\mathcal F}_\varepsilon$ is nonincreasing,
tends to $+\infty$ as $a\downarrow0$, and is bounded as $a\to\infty$.
Consequently $a_\varepsilon^{\rm c}$ in \eqref{eq:coupled-crossing} is finite.
\end{lemma}

\begin{proof}
Monotonicity in $a$ follows from the monotonicity of both survivor lower
bounds and the nesting of the feasible intervals.  As $a\downarrow0$,
$y(a)\uparrow1$.  For the fixed choice $\lambda=1/2$, both
$\phi_{\varepsilon,1/2}(a)$ and $\psi_{1/2}(m,a)$ converge to
$\varepsilon$ uniformly for $m\in[y(a),1]$.  Since
$1-\varepsilon m\ge1-\varepsilon$, every finite term in the infimum
therefore diverges, and terms with $B_{\varepsilon,1/2}\ge\varepsilon$ are
already $+\infty$.  Thus $\widetilde{\mathcal F}_\varepsilon(a)\to+\infty$.
As $a\to\infty$, take $m=1$ in the infimum.  The scalar minimization giving
$\phi$ implies $\phi_{\varepsilon,\lambda}(a)\le\varepsilon y(a)$ for every
$\lambda$, while
$\psi_\lambda(1,a)=\varepsilon 2^{-a/(1-\lambda)}\le\varepsilon2^{-a}$.
Consequently the resulting rate is bounded by
\[
 \sup_{0<\lambda<1}\lambda\beta
 \left[\log_2\frac{\beta}{\varepsilon}\right]_+<\infty.
\]
Thus the threshold is finite.
\end{proof}

\begin{proof}[Proof of \Cref{thm:general-epsilon}]
Fix $0<a<a_\varepsilon^{\rm c}$.  By
\Cref{lem:coupled-balance},
$\widetilde{\mathcal F}_\varepsilon(a)>\ell+a$.  Choose a fixed $\lambda$ and
$\gamma>0$ such that, with
\[
 g_{\varepsilon,\lambda}(a,m)=
 \lambda(1-\varepsilon m)
 \left[\log_2\frac{1-\varepsilon m}
 {\varepsilon-B_{\varepsilon,\lambda}(a,m)}\right]_+,
 \]
\[
 \inf_{y(a)\le m\le1}g_{\varepsilon,\lambda}(a,m)
 \ge \ell+a+3\gamma. \tag{8.2}\label{eq:strict-margin}
\]
Suppose, toward a contradiction, that $H/n\le\ell+a$ along a sequence with
$u/n\to\infty$.

The Carter decomposition, the joint posterior chain rule, and
\Cref{lem:coupled-survivor} give
\[
\E[\mathcal E_\varepsilon(X)+\kappa]\le a+o(1),
\qquad
m=\E X/\varepsilon\in[y(a),1]+o(1),
\]
and
\[
\E\left[X2^{-\kappa/(1-\lambda)}\right]
\ge B_{\varepsilon,\lambda}(a,m)-o(1).
\]
Let $\widehat m$ be the projection of $m$ onto $[y(a),1]$.  Then
$|m-\widehat m|=o(1)$, and all nonsingular terms in the last display vary by
$o(1)$.  If
$B_{\varepsilon,\lambda}(a,\widehat m)>\varepsilon$, the average reservoir
bound is itself impossible for large $n$.  If equality holds, it gives
$\E c=o(1)$; since $\E(1-X)\ge\beta-o(1)$, the raw reservoir entropy
inequality forces $H/n\to\infty$, again a contradiction.  Otherwise the
strict margin in \eqref{eq:strict-margin} keeps the denominator positive and
uniformly below the numerator on the relevant compact subinterval.  The
coupled reservoir lemma then applies with
$c_0=\varepsilon-B_{\varepsilon,\lambda}(a,\widehat m)$ and yields
\[
\frac Hn\ge
g_{\varepsilon,\lambda}(a,\widehat m)-o(1)
\ge\ell+a+2\gamma-o(1),
\]
contradicting the assumed space bound.  Since every
$a<a_\varepsilon^{\rm c}$ is excluded, the theorem follows.
\end{proof}

\section{Discussion}\label{sec:discussion}

The proof separates three resources that can be conflated in dynamic lower
bounds: accepted-set volume, posterior nonuniformity inside that set, and
transition width after a common update word.  Its main point is that these
resources need not be controlled by separate thresholds.  The tangent gap
$\mathcal E_\varepsilon(X)$ and branch cost $\kappa$ form one energy, while
the successor reservoirs are charged by conditional entropy on average.

The joint conditioning on $(D,I)$ is essential.  An argument that first fixes
a witness source and then asks one insertion batch to avoid all alternative
witnesses incurs collision probability of order $n^2/u$.  Here each branch
has its own compatible posterior support.  The relative-entropy chain rule
charges all posterior loss before any tape or source is fixed.  The reservoir
entropy lemma then keeps the correlation between an insertion batch and its
successor state instead of discarding it through a union bound.

The threshold $a_\varepsilon^{\rm c}$ is determined by a coupled balance between
the Carter-energy loss and integrated replacement covering.  The local
profile $\phi_{\varepsilon,\lambda}$ records what the total energy alone
forces, while $\psi_\lambda(m,a)$ additionally records the mean accepted
mass.  The final envelope combines both constraints with the same reservoir
entropy cost.  The proof contains no cutoff or good-branch pruning.  Further
improvement would require using more of the conditional posterior than its
support union, or coupling several live successor response tables rather than
only one successor state.

The argument is a one-step converse for fixed worst-case persistent space.  It
applies to arbitrary legal histories and makes no assumption about update or
query time, stationarity, monotonicity, or the internal representation.  The
bound is therefore a genuine lower bound for the ordinary model; determining
the matching optimum, or extending the argument to a finite-horizon
distributional formulation, remains open.

\section*{Acknowledgments}

The authors acknowledge GPT-5.6 Sol, accessed through OpenAI Codex, and
DeepSeek V4 Pro, accessed through DeepSeek Harness, for assistance with
exploring and developing proof strategies, literature organization,
adversarial proof checking, \LaTeX{} preparation, and editorial revision.

\appendix
\section{Technical details}\label{sec:appendix-details}

This appendix makes explicit the treatment of the public tape and asymptotic
quantifiers.

\subsection{The reference conditional law in full detail}

Fix $z=(r,m)$ and let $A=A_z$.  Under the reference source law, a feasible
triple $(d,i,k_0)$ satisfies
\[
 |d|=|i|=q,\quad |k_0|=k,\quad
 d\cap k_0=d\cap i=k_0\cap i=\varnothing,
 \quad d\cup k_0\subseteq A.
\]
The source is necessarily $d\cup k_0$.  Its probability is $1/\binom an$;
the deletion set then has probability $1/\binom nq$ and the insertion set has
probability $1/\binom{u-n}q$.  This proves that the joint probability is
constant over all feasible triples.  For fixed $(d,i)$, feasibility of $k_0$
is exactly
\[
  k_0\in\binom{A\setminus(d\cup i)}k.
\]
This proves \eqref{eq:reference-conditional} without making any assertion
about the reference marginal distribution of $(D,I)$.

\subsection{Public tapes and conditional entropy}

The random tape can be modeled as a probability space
$(\mathcal R,\mathcal B,\Prb_R)$, possibly nonatomic.  For a fixed tape $r$,
the parent experiment has only finitely many source sets and states.  The map
$s\mapsto M_r(s)$ therefore defines the finite fibers $\cF_{r,m}$ used in
\Cref{sec:model}; no regular conditional probability at the singleton event
$R=r$ is required.

All expectations in \Cref{lem:carter,lem:joint-kl,lem:reservoir-entropy} may
first be evaluated for fixed $r$ over these finite fibers and then integrated
over $r$.  Conditional entropies involving the finite state $M'$ are therefore
ordinary expectations of finite-alphabet entropies.  The false-positive
guarantee is used only through the unconditional accepted-size averages in
\Cref{lem:fpr-average}; the proof never conditions that guarantee on a
selected tape.

\subsection{Floors and asymptotics}

For every fixed $0<\lambda<1$, the choices
$q=\lfloor\lambda n\rfloor$ and $k=n-q$ satisfy
\[
  q/n=\lambda+O(1/n),\qquad
  k/n=1-\lambda+O(1/n).
\]
They therefore change every normalized entropy estimate by only $o(1)$.

Every occurrence of $o(1)$ or $o(n)$ in the main proof is along an arbitrary
sequence $(n_j,u_j)$ satisfying $n_j\to\infty$ and $u_j/n_j\to\infty$, with
$\varepsilon$ fixed.  In the contradiction argument, $a$ and the selected
$\lambda$ are also fixed before the asymptotic limit is taken.

\bibliographystyle{plainnat}
\bibliography{references}

\end{document}